\documentclass[prodmode,acmtosn]{acmsmall}

\usepackage[ruled]{algorithm2e}
\usepackage{graphicx}
\usepackage{supertech-sig}

\newcommand{\suba}{\mbox{${\cal F}_{sa}$}\xspace}

\SetAlFnt{\small}
\SetAlCapFnt{\small}
\SetAlCapNameFnt{\small}
\SetAlCapHSkip{0pt}
\IncMargin{-\parindent}

\newcommand{\old}[1]{{}}

\begin{document}

\markboth{M.~E.~Bender, M.~Farach-Colton, S.~P.~Fekete, J.~Fineman, and S.~Gilbert}{Cost-Oblivious Storage Reallocation}

\title{Cost-Oblivious Storage Reallocation}

\author{
    Michael~A.~Bender
\affil{Stony Brook University}
    Mart{\'i}n~Farach-Colton
\affil{Rutgers University}
    S{\'a}ndor~P.~Fekete
\affil{TU Braunschweig}
    Jeremy~T.~Fineman
\affil{Georgetown University}
    Seth~Gilbert
\affil{National University of Singapore}
}

\begin{abstract}
Databases allocate and free blocks of storage on disk.  Freed blocks
introduce holes where no data is stored.  Allocation systems attempt
to reuse such deallocated regions in order to minimize the footprint
on disk.  When previously allocated blocks cannot be moved, this
problem is called the \defn{memory allocation} problem.  The
competitive ratio for this problem has matching upper and lower bounds
that are logarithmic in the number of requests
and in the ratio of the largest to smallest requests.

This paper defines the \defn{storage reallocation} problem, where
previously allocated blocks can be moved, or \defn{reallocated}, but
at some cost. This cost is determined by the allocation/reallocation
\defn{cost function}.


The objective is to minimize the storage footprint, that is, the
largest memory address containing an allocated object, while
simultaneously minimizing the reallocation costs.  This paper gives
asymptotically optimal algorithms for storage reallocation, in which
the storage footprint is at most $(1+\epsilon)$ times optimal, and the
reallocation cost is $O((1/\epsilon) \log(1/\epsilon))$ times
the original allocation cost, that is, it is within a constant factor
of optimal when $\epsilon$ is a constant.
The algorithms are \defn{cost oblivious}, which
means they achieve these bounds with no knowledge of the
allocation/reallocation cost function,  as long as the cost function
is subadditive.
\end{abstract}




\begin{bottomstuff}
A previous extended abstract version appears in the Proceedings of the 
33rd~ACM SIGMOD-SIGACT-SIGART Symposium on Principles of Database Systems, PODS '14~\cite{BenderFaFe16}.
This research was supported in part by NSF grants IIS~1247726, CCF~1217708, CCF~1114809,
CCF~0937822, and CCF~1617618, and  Sandia National Laboratories
(Michael A.\ Bender), NSF grants IIS~1247750 and
CCF~1114930 (Mart{\'i}n Farach-Colton), by DFG grant FE407/17-1 and 17-2, as part of the
Research Group FOR~1800, ``Controlling Concurrent Change" (S{\'a}ndor P.
Fekete), by NSF grant CCF~1218188 (Jeremy T.\ Fineman), by MOE Tier 2 Grant MOE2014-T2-1-157
(Seth Gilbert).

Author's addresses:

Michael~A.~Bender, Department of Computer Science, Stony Brook University,Stony Brook, NY 11794-2424,  USA. bender@cs.stonybrook.edu.
Mart{\'i}n Farach-Colton, Department of Computer Science, Rutgers University, Piscataway, NJ 08854, USA. farach@cs.rutgers.edu.
S\'andor~P.~Fekete, Department of Computer Science, TU Braunschweig, 38106 Braunschweig, Germany.  s.fekete@tu-bs.de.
Jeremy~T.~Fineman, Department of Computer Science, Georgetown University, Washington, DC 20057, USA. jfineman@cs.georgetown.edu.
Seth~Gilbert, Department of Computer Science, National University of Singapore, Singapore 117417, Singapore. seth.gilbert@comp.nus.edu.sg.

\end{bottomstuff}

\maketitle


\section{Introduction}
\seclabel{intro}
\sloppy

Databases, and more generally storage systems, need to allocate and
free blocks of storage on disk.  Freed data introduces holes where no
data is stored.  Allocation systems attempt to reuse such 
deallocated regions in order to minimize the footprint on disk.

The problem of allocating and freeing storage is well studied as the
\defn{memory allocation} problem.  In that formulation, allocated
objects cannot be moved.  The \defn{competitive ratio} is defined to be
the maximum possible ratio of the allocated memory (largest allocated
memory address) to the sum of the sizes of allocated
segments~\cite{Knuth97,LubyNaOr96,NaorOrPe00}.  The lower bound on
the competitive ratio is logarithmic in the number of requests
and in the ratio of the largest to smallest request~\cite{LubyNaOr96}.

The logarithmic lower bound renders traditional memory allocation too
blunt a theoretical tool for understanding storage in many settings.
Furthermore, as we show, this lower bound is a consequence of the
requirement that allocated storage cannot be moved.  But many actual
systems have no such restriction.

\begin{figure}[hbt]
  \begin{center}\includegraphics[scale=.4]{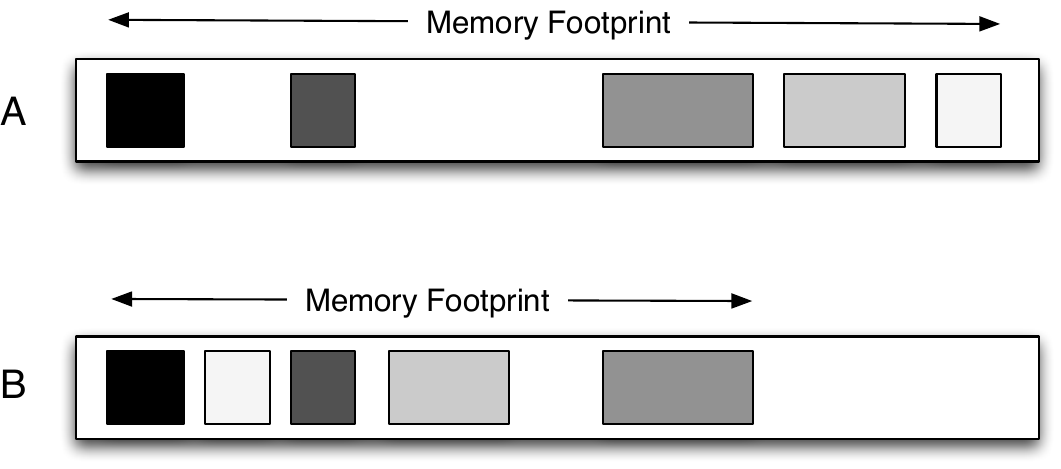}\end{center}
  \caption{Moving previously allocated blocks into holes left by
    deallocations can reduce the footprint of the data in storage.}\figlabel{layout}
\end{figure}

\subheading{Storage reallocation}  This paper generalizes memory
allocation by allowing the allocator to move previously allocated
objects.  We call this generalization \defn{storage reallocation}.
Storage reallocation can take place on any physical medium for allocating
objects, e.g., main memory, rotating disks, or flash memory.  See \figref{layout}.

Thus, garbage 
collection~\cite{JonesHoMo11}
is a type of in-core storage reallocation.
More generally,  systems that introduce a layer of indirection between
logical addresses and physical addresses, such
as virtual memory,
make reallocation
transparent to processes that  request storage.

Our own interest in memory
reallocation stems from our experience in building the
TokuDB~\cite{TokuDB} and TokuMX~\cite{TokuMX} databases, in which
memory segments are accessed via a so-called ``block translation
layer,'' which translates between the block name, which is immutable,
and the block address in storage, which may change.
{\color{black}{(While TokuDB often reallocates storage, 
its reallocator does not enjoy the extra property of cost-oblivousness 
addressed in this paper.)}} 



\subheading{Cost-oblivious storage reallocation} 
An algorithm for storage reallocation must contend with the tradeoff
between storage footprint size and the amount (and cost) of reallocation.  
It should come as no surprise that a storage reallocator that is
designed for main memory is unlikely to work well if the objects are
allocated on a rotating device instead---and vice versa.  This is
because the cost model depends on where the objects are stored.

The question is therefore how to model the cost of reallocating memory
objects.  Faithful cost models are hard to come by, in part because the memory
hierarchy has a hard-to-quantify impact on run time.  In RAM, moving
an object is roughly proportional to the object size.  On disk,
moving a small object may be dominated by the seek time, while moving
a large object may be dominated by the disk bandwidth.  In both cases,
there are cache effects, both in memory and in storage and in their
interaction.  The performance
characteristics for each aspect of memory vary by brand and model.

Rather than model these complex interactions, this paper specifies a class of
cost functions that subsumes them.  We give universal reallocators,
independent of the particulars of the reallocation cost.  We say that
a universal reallocator is \defn{cost oblivious} with respect to
a class of cost functions if its execution does not depend on the
specific choice of 
cost function from the class.  
Our  reallocation algorithms are cost oblivious with respect 
to the class
of cost functions that are
subadditive, monotonically
increasing functions of the object size. (A (monotonically
  increasing) function $f(x)$ is \defn{subadditive}, if $f(x+y)\leq
  f(x) +f(y)$ for any positive $x$ and $y$.  Note that all monotonically
  increasing concave functions are subadditive.)  
  The restriction to subadditivity is not severe. While there exist
  corner cases where a storage system is temporarily superadditive,
  most mechanisms employed by operating systems, such as prefetching
  for latency hiding, rely on the subadditivity of costs. 

To summarize, in  storage reallocation, there is an online sequence of insert
(memory allocation, i.e., function {\em malloc}) and delete (memory release, i.e., function {\em free}) requests.  Objects are allocated to
locations in an arbitrarily large array (address space).  The cost of allocating or 
moving (reallocating) a size-$w$ object is some unknown
(monotonically increasing) subadditive function $f(w)$.

Storage reallocation is thus a bicriteria optimization problem. 
The first objective is to store
objects in an array so that the largest allocated
memory address---which we call the \defn{footprint}---is
approximately minimized.
The second objective is to minimize the 
amortized reallocation cost per new request.  
In this paper, we consider the problem of minimizing the amortized reallocation cost, while using a 
memory footprint that is at most a constant factor larger than optimal.

\subheading{Storage reallocation in a database}  Databases have many
moving parts, and any system that changes the way that storage is
allocated needs to interact gracefully with the other requirements of
the storage system. 

A common constraint in storage (re)allocation is that updates be
\defn{nonoverlapping}, i.e., when an object is
moved, its new location must be disjoint from its old location.  
In databases, object writes are not atomic, so nonoverlapping
reallocation is necessary for durability. This is also relevant in other contexts: In SSDs, the nonoverlapping constraint is enforced by the
hardware, because memory locations must be erased
between writes.
In
FPGAs, satisfying this constraint allows interruption-free reallocations of 
modules~\cite{FeketeKaSc12}.    

The nonoverlapping constraint is only part of the mechanism for
durability.  Another consideration is that when an object is moved,
the translation table between logical and physical addresses needs to be
updated.  It is then written to disk during a \defn{checkpoint}.  Only then
are blocks that have been freed since the last checkpoint available for
reuse.  Therefore, the allocator may not write to a location that has
been freed after the last checkpoint.

Finally, new memory requests arrive at unpredictable times.  It is
undesirable for an allocation request to block on a long sequence of
reallocations, even if the average throughput is high.  A good reallocation 
algorithm should provide some guarantee on the worst-case cost of
individual operations,  
while still maintaining (near) optimal throughput.   

\subheading{Formalization} 
An \defn{online execution} is a sequence of requests of the form
$\ang{\proc{InsertObject}, \id{name}, \id{length}}$ and
$\ang{\proc{DeleteObject}, \id{name}}$.  After each request, the
reallocator outputs an allocation for the objects in the
system.  We say that an object is \defn{active} at time $t$ if it 
has been inserted by one of the first $t$ requests, but not deleted 
by the end of request $t$.  (Note that an object being deleted remains active
until the reallocator completes the delete request.)

If $\cal S$ and $\cal S'$ are the allocations immediately
before and after request $p$, then the \defn{reallocation cost of $p$}
is the sum of the reallocation costs of all objects moved between
$\cal S$ and $\cal S'$.

A reallocator $A$ is \defn{$(f, a, b)$-competitive} for cost function
$f()$, if (1) the footprint size is always optimized to within an
$a$-factor of optimal, and (2) the reallocation cost is at most $b$
times the sum of the allocation costs of every object inserted so
far (including those that have subsequently been deleted).  Since every 
object must be allocated at least once, the cost of such a reallocator is 
within a factor of $b$ of optimal.

Let $\cal C$ be a set of cost functions. A reallocation algorithm $A$
is \defn{cost oblivious} if it does not depend on $f()$. This
  means not only that $f()$ is not a parameter to algorithm $A$, but
  also $A$ learns nothing about $f()$ as $A$ executes. A
cost-oblivious reallocator $A$ is \defn{$({\cal C}, a,
  b)$-competitive} if it is $(f, a, b)$-competitive for every $f \in
{\cal C}$; we abbreviate to \defn{$(a,b)$-competitive} if the set
$\cal C$ is unambiguous.  
In the remainder of this paper, we take ${\cal C}=\suba$, the class of 
monotonically increasing, subadditive functions.

\subheading{Results} Our reallocation algorithms are tunable to
achieve an arbitrarily good competitive ratio $1+\epsilon$ ($0 <
\epsilon \leq 1/2$)
with respect to the footprint size.  
All objects have integral length, and $\Delta$ denotes the
length of the longest object. We establish the following: 

\begin{itemize}
\item We give a cost-oblivious algorithm for storage reallocation that
  is $(\suba,1+\epsilon,O((1/\epsilon)\log(1/\epsilon))$-competitive.
  This allocator is amortized in the sense that it might reallocate
  every existing object between servicing two requests.

\item As a corollary, we give a defragmenter that is cost oblivious
  with respect to $\suba$.  The defragmenter takes as input a
  comparison function, a set of objects having total length $V$ and
  consuming space $(1+\epsilon)V $.  The defragmenter sorts the
  objects using $(1+\epsilon)V +\Delta$ working space, moving each
  object $O((1/\epsilon)\log (1/\epsilon))$ times, {\color{black}{amortized}}. 

\item We extend the storage reallocator to support checkpointing.
  With an additional $O(\Delta)$ space, we guarantee that each
  operation completes within $O(1/\epsilon)$ checkpoints.  
  
\item We also partially deamortize the storage reallocator so that the
  worst-case reallocation cost (and therefore the worst-case time
  blocking for a new size-$w$ allocation) is reduced to
  $O((1/\epsilon)wf(1) + f (\Delta))$.
\end{itemize}

{\color{black}{There is a variety of possible extensions to this concept.
One such direction is to consider the sum of allocation costs;
we address this in a related followup paper~\cite{bff+-corsct-15}.}}

\subheading{Related work} We now  review the related work.

\noindent \emph{Dynamic memory allocation.} There is an extensive
literature on memory
allocation~\cite{Knuth97,Robson71,Robson74,Robson77,LubyNaOr96,NaorOrPe00,Woodall74}
where object reallocation is disallowed.  There are upper and lower
bounds on the competitive ratio of the memory footprint that are
roughly logarithmic in the number of requests and in the ratio of the
largest to smallest request.  These papers generally analyze
traditional strategies such as Best Fit, First Fit, and the Buddy
System~\cite{Knowlton65}, but also propose alternatives.
Traditional memory-allocation strategies often have analogs in
bin-packing~\cite{CoffmanGaJo83,CoffmanJoSh93,CoffmanJoSh97,CoffmanGaJo96,GalambosWo95},
but an enumeration of such results lies beyond the scope of this
paper.

Memory allocation where reallocation \emph{is allowed} appears often in the
literature on garbage collection~\cite{JonesHoMo11}. 
There is a long and important line of literature studying dynamic memory
allocation with differing compaction mechanisms, exploring the 
time/space trade-off between the 
 amount of compaction performed 
and the total memory used.  
Ting~\cite{Ting76} develops a mathematical model for examining this trade-off
for different compaction algorithms; B\l{}a\.{z}ewicz et
al.~\cite{BlazewiczNa85} develop a ``partial'' compaction algorithm for
segments of two different sizes that reallocates only a limited number of
segments per compaction. More recently Bendersky and Petrank~\cite{BenderskyPe12} and Cohen and Petrank~\cite{CohenPe13} have more fully
explored the trade-offs inherent in partial compaction.  

These papers on dynamic memory allocation with compaction are
instances of storage reallocation, as addressed in this paper, where
the reallocation cost is (typically) linear: the cost of compaction is
directly proportional to the amount of memory that is moved.  (These
papers often address other problems that arise in garbage collection,
such as how to update pointers to memory that has moved.)  For
example, Bendersky and Petrank~\cite{BenderskyPe12} show that when the
cost function is linear, one can achieve constant amortized
reallocation cost with memory size that is within a constant-factor of
optimal.


In this paper, by contrast, we focus on cost-oblivious algorithms that
tolerate the range of cost functions found in external storage
systems.  Cost obliviousness bears a passing resemblance to similar
notions in the memory hierarchy, particularly the
cache-oblivious/ideal-cache~\cite{FrigoLePr99,Prokop99}, hierarchical
memory~\cite{AggarwalAlCh87}, and
cache-adaptive~\cite{BenderEbFi14,BenderDeEb16} models.  With the
exception of the underlying paging~\cite{SleatorTa85}, work in these
models is about writing algorithms that are memory-hierarchy universal
rather than analyzing resource allocation. {\color{black}{Although we
    consider an online setting, even finding optimal offline
    algorithms seems nontrivial. }}

\noindent \emph{Other related work.}
Storage reallocation has other applications besides
databases.  For example, Fekete et al.~\cite{FeketeKaSc12} address the
storage reallocation problem in the context of FPGAs, and Bender et
al.~\cite{BenderFeKa09} give (not cost-oblivious) algorithms for
constant reallocation cost.

Sparse table data structures
\cite{ItaiKoRo81,Willard82,Willard86,Willard92,Katriel02,BenderDeFa05,ItaiKa07,BenderCoDe02a,BulanekKoSa12,BenderHu07,BenderHu06,BenderBeJo16,BenderBeJo16,BenderFaMo06,BenderFiGi17}
also solve the storage reallocation problem and are easily adapted to
deal with different-sized objects and linear reallocation cost.  But
they do so while maintaining the constraint that the object order does
not change, which makes the problem harder and the reallocation cost
correspondingly larger.

\noindent \emph{Scheduling/planning interpretation.}  The storage reallocation
problem can be viewed as a reallocation problem in
scheduling/planning.  In this interpretation, we have an online
sequence of requests to insert a new job $j$ into the schedule or to
delete an exiting job $j$. Each job has a length $w_j$ and the
rescheduling cost is $f(w_j)$.  The goal is to maintain a uniprocessor
schedule that (approximately) minimizes the makespan (latest
completion time of any job), while simultaneously guaranteeing the
overall reallocation cost is approximately minimized.  We can
abbreviate this scheduling problem as $1| f(w) \mbox{\it ~realloc
}|C_{\rm max}$, generalizing standard scheduling
notation~\cite{GrahamLaLe79}.  The goal is actually not to \emph{run}
the schedule, but rather to plan a schedule subject to an online
sequence of changes to the scheduling instance.

We thus  review related work in scheduling and combinatorial
optimization.  Several papers explore related notions of scheduling
reallocation (although to the best of our knowledge, not
cost-universal scheduling reallocation).  Bender et
al.~\cite{BenderFaFe13} study reallocation scheduling with unit-length
jobs having release times and deadlines.  Their reallocator maintains
a feasible multiprocessor schedule while servicing inserts and
deletes.

In the area of robust optimization, the goal is to develop solutions for combinatorial optimization problems that are (near) optimal, and that can be readily updated if the instance changes.  In this context, many papers have looked at the problem of minimizing reallocation costs for specific optimization problems (e.g.,~\cite{HallPo04,UnalUzKi97,FeketeKaSc12}).  For example, 
Davis et al.~\cite{DavisEdIm06} study a reallocation problem, where an
allocator divides resources among a set of users, updating the
allocation as the users' constraints change.  The goal is to minimize
the number of changes to the allocation.  
%
%
As another example, Sanders et al.~\cite{SandersSiSk09} look at the problem of assigning jobs to processors, minimizing the reallocation as new jobs arrive.   Jansen et al.~\cite{JansenKl13} look at robust algorithms for online bin packing that minimize migration costs.  See Verschae~\cite{Verschae12} for more details on robust optimization.

Shachnai et al.~\cite{ShachnaiTaTa12} explore a slightly different notion of reallocation
for combinatorial problems.  Given an input, an optimal solution
for that input, and a modified version of the input, they develop
algorithms that find the minimum-cost modification of the optimal
solution to the modified input.  A difference between their setting
and ours is that we measure the ratio of reallocation cost to
allocation cost, whereas they measure the ratio of the actual
transition cost to the optimal transition cost resulting in a good
solution.  Also, we focus on a sequence of changes, which means we
amortize the expensive changes against a sequence of updates.

There also exist reoptimization problems, which address the
goal of minimizing the computational cost for incrementally updating
the
schedule~\cite{AusielloBoEs11,ArchettiBeSp10,AusielloEsMo09,ArchettiBeSp03,BockenhauerFoHr06}.
By contrast, in reallocation, we focus on the cost of reallocating
resources rather than the computational cost of generating the
allocation.

\newcommand{\mmax}{\Delta}

\section{\sloppy Footprint Minimization}
\seclabel{simplealloc}

In this section we give a cost-oblivious algorithm for footprint
minimization in storage reallocation.  The footprint always has size
at most $(1+\epsilon) V_{t}$, where $V_{t}$ denotes the
\defn{volume}, or total size, of all allocated objects at time $t$,
i.e., of the active objects after the $t$th operation completes.  A
size-$w$ object has an amortized reallocation cost of
$O\left(f(w)\cdot(1/\epsilon)\log(1/\epsilon)\right)$, where $f(w)$ is
the (unknown) cost for allocating an object of size $w$.

\begin{theorem}
  For any constant $\epsilon$ with $0 < \epsilon \leq 1/2$, there
  exists a cost-oblivious
  storage-reallocation algorithm that is $\left(1+\epsilon,\,
    O\left((1/\epsilon)\log(1/\epsilon)\right)
  \right)$-competitive with respect to \suba, the class of monotonically increasing, subadditive cost functions.
  \thmlabel{makespan}
\end{theorem}
Thus, the storage reallocation algorithm 
is within a constant factor of optimal for any constant $\epsilon<1/2$.

\subheading{Intuition and cost-function-specific algorithms}
We begin by considering some simple cases where the cost function is
known in advance.  First suppose that the reallocation cost is linear
in the object size, i.e., $f(w)=w$.  A simple logging-and-compressing
strategy attains a $(2,2)$-competitive algorithm for linear cost
functions.  Specifically, allocate objects from left to right.  Upon a
deletion, leave a hole where the object used to be. Whenever a
deallocation causes the footprint to reach $2V_{t}$, remove all holes
by compacting. The cost to reallocate the entire volume $V_t$ is paid for
by the $V_t$'s worth of elements that were deallocated since the last
compaction.

Logging and compressing does not work well for constant reallocation
cost, i.e., $f(w)=1$.  To see why, suppose the deleted objects have
size $\Delta$, and the reallocated elements have size~$1$.  We may
need to spend amortized $\Theta(\Delta)$ reallocation cost per deletion.

There do exist good reallocators for constant 
reallocation cost~\cite{BenderFeKa09}. 
Conceptually, round the object sizes up 
to the next power of $2$ to form \defn{size classes}, 
where 
objects have  size 
$2^{i}$ for $i=0,\ldots,\log{\Delta}$. 
Now group the objects by increasing size. 
Between the $i$th and $(i+1)$st size class, there 
is either a gap of size  $2^{i}$ or 
no gap. To insert an object of size $2^{i}$, 
put the object into the  gap after the $i$th size class, if one exists, 
or displace a larger object to make space, otherwise. 
Then recursively reinsert the larger object. 
The amortized reallocation cost is $O(1)$, 
because the costs per unit volume 
to displace the recursively larger objects form a geometric series.  

It can be shown, however, that with linear reallocation cost this 
strategy is only $(2,{\color{black}{\Theta}}(\log \Delta))$-competitive.

This section gives a single algorithm that works for $f(w)=w$,
$f(w)=1$, and all other subadditive cost functions.  The algorithm
keeps the objects partially sorted by size.  Since the cost function
is subadditive, small objects are the most expensive to move per unit
size.  We therefore want to guarantee that when an object is inserted
or deleted, it can only trigger the movement of larger (less expensive
per unit size) objects.  Specifically, small objects with total volume
$W$ will be able to cause the movements of big objects with
total volume $O(W)$, but not the other way around.  At the same time,
we need to avoid cascading reinserts, which can happen with the 
algorithm for unit cost described above.

\begin{figure}
  \begin{center}\includegraphics[scale=.4]{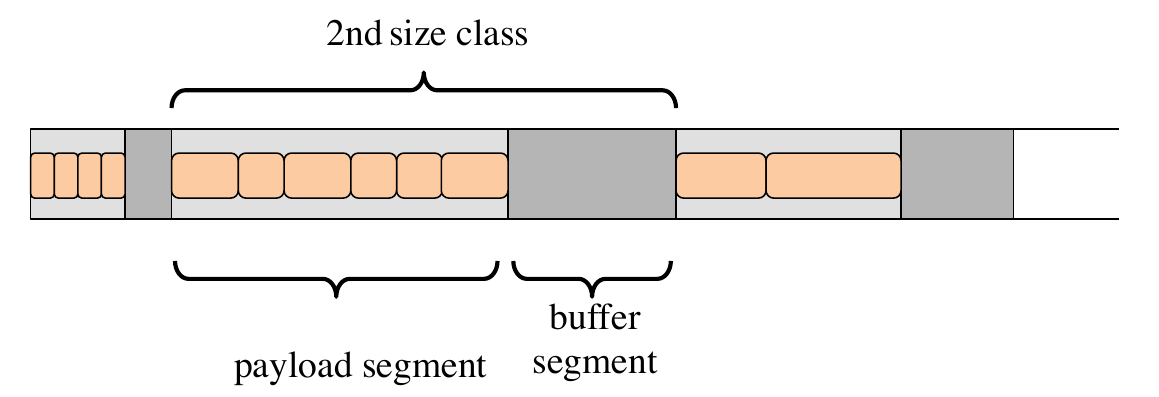}\end{center}
  \caption{The layout of the data structure when the buffer segments
    are empty, with $\epsilon' = 1/2$. The light-gray are the
    payload segments, and the dark-gray areas are the buffer
    segments.  The orange rectangles are objects currently in the data
    structure.}\figlabel{layout-class}
\end{figure}

\subheading{Overview and invariants}
Objects are categorized into \defn{size classes}; 
the~$i$th size class contains
objects of size $w$, where $2^{i-1} \leq w < 2^i$.
Thus, there are 
$\floor{\log\mmax}+1$ \defn{size classes}. 
The value
of $\Delta$ need not be known in advance.    For size class~$i$,
$V_t(i)$ denotes the \defn{volume} (total size) of all objects active
at time $t$ in size class $i$.  If $t$ is understood, we use $V(i)$.

{\color{black}{Intuitively, these size classes allow us to order objects
    by approximate size, which helps make efficient deletes possible.
    To maintain our target makespan, we need to reallocate an object
    when too many objects to its left are deleted. If objects to the
    right are small and objects to the left are large, then
    reallocations are too expensive for most cost functions. Within a
    size class, the ordering does not matter, since it only affects
    the reallocation cost by a constant factor.
    We next explain that, in fact, we can even further relax our ordering. 
 }}

The array (address space) is divided into $\floor{\log\mmax}+1$
regions, as illustrated in \figref{layout-class}.  The $i$th region is
dedicated to the $i$th size class and comprises two subregions, a
\defn{payload segment} followed by a \defn{buffer segment}.  The $i$th
payload segment contains only objects belonging to the $i$th size
class, whereas the $i$th buffer segment may contain objects that are
in the $i$th size class or smaller size classes.

Whenever (potentially large) reallocations are taking
place, an \defn{overflow segment} is used 
for temporarily
rearranging the objects, as described later.
The overflow segment is placed at  the end of the array.

\begin{invariant}
The following properties are maintained throughout the 
execution of the algorithm: 
\begin{enumerate}
\item The $i$th region ($i=1,\ldots,\floor{\log\mmax}+1$) 
comprises the $i$th payload and $i$th buffer segment. 
\item The $(\floor{\log\mmax}+2)$nd region, the overflow segment, 
stores elements temporarily during reallocation. 
\item The $i$th payload segment only stores elements 
from the $i$th size class.
\item The $i$th buffer segment only stores elements 
from size classes $\ell\leq i$. 
\end{enumerate}
\end{invariant}

\subheading{Allocating and deallocating}
When a new size-$w$ object that belongs to a size class $i$ is
allocated, it is stored at the end of the earliest buffer $j \geq i$
that has sufficient unoccupied space.  (Recall that this object cannot 
be inserted into any buffer in a segment less than $i$.)

When there is not enough available space in any of these
buffers, a \defn{buffer flush} operation is triggered (see \figref{flush}), after which
the object is inserted. 
During a buffer flush, all objects in some suffix
of buffers get moved to their proper payload segments and the segment and region boundaries get redefined. 

If the new size-$w$ object belongs to a larger size class than any
other active object, then we instead create a new payload segment and
buffer segment for the new size class located immediately after the
last size class's segment, increasing the total space used by at most
an additive $w + \epsilon' w$, for some constant $\epsilon'$.  (The
overflow segment is empty, because it is only used during a buffer flush,
and hence  implicitly resides after the new size class.)

When a size-$w$ object is deleted, it leaves a hole until the
next buffer flush occurs.  A dummy deletion request is added to the
buffer and forced to consume $w$ space.  This buffered dummy request is not
freed until the next buffer flush.  Since both inserting and deleting a job
of size $w$ reduces the space in the buffer by $w$, we can analyze
insertions and deletions together.

\begin{invariant}
The overflow segment is empty except during buffer flush operations. 
\end{invariant}

\begin{invariant}
When a flush of the $i$th buffer segment occurs at time
$t$, the object and segment boundaries move so that:
\begin{enumerate}
\item  the space occupied by the $i$th payload segment after
the buffer flush completes is exactly $V_t(i)$, and 
\item the space occupied by the $i$th
buffer segment after the buffer flush completes is $\floor{\epsilon' V_t(i)}$, for $\epsilon' =
\Theta(\epsilon)$. 
\end{enumerate}
Immediately following this
flush, the size-$\floor{\epsilon' V_t(i)}$ buffer contains no
objects.
\end{invariant}

{\color{black}{
As described in this
  section, $V_t(i)$ immediately increases to count the new object, but
  the object is not yet placed in the array.  Next, the flush occurs,
  and finally the new object is placed in the array.  Our extension in
  \secref{db} places the object before performing the flush; this
  extension requires an additive $\Delta$ working space during the
  flush procedure. 
}}

\begin{figure}
  \begin{center}\hspace{-1mm}\includegraphics[scale=.385]{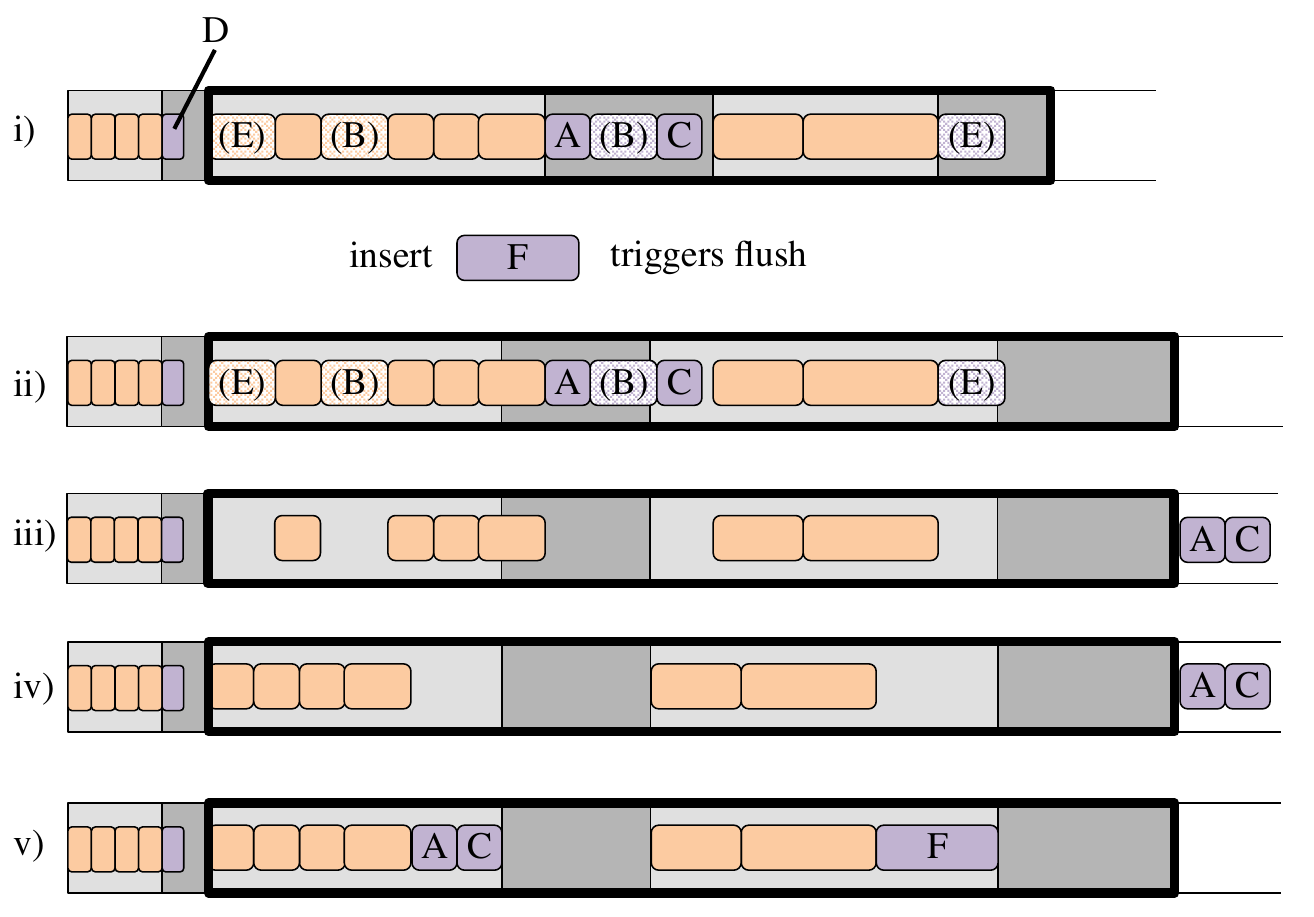}\end{center}
  \caption{Example of a flush, starting from \figref{layout-class}.  The
    lavender rectangles are updates to the data structure, with
    parentheses and light shading denoting a delete or delete
    record. (i) The state after insert A, delete B, insert C, insert
    D, and delete E in that order. (ii)--(v) show a flush that occurs
    when inserting F.  The heavy outline shows the region affected by
    the flush, i.e., size classes 2 and 3. (ii) The new boundaries for
    the 2nd and 3rd size class.  (iii) The state after moving buffered
    objects out of the way and dropping deleted objects.  (iv) The
    state after rearranging the payload segments. (v) The state after
    putting all buffered objects to their proper locations.  Observe
    that for the flushed classes, the buffers are now
    empty.}\figlabel{flush}
\end{figure}

\subheading{Buffer flush} A buffer flush updates the segment
boundaries in a suffix of regions, moving all objects to their proper
payload segments, and leaving all buffer segments empty to accommodate
future insertions.

To execute a buffer flush, first determine the \defn{boundary size
class} $b$ and then flush all buffers for size classes $i \geq b$.
The value $b$ is defined as the maximum value such that all objects in
buffers $i\geq b$ and the object being inserted/deleted belong to size
classes at least $b$.  To determine $b$, iterate from the largest to the
smallest region, examining every object in the region's buffer.  If
any object belongs to a size class $s < b$, then update $b$ with the
size class $s$.  This continues until reaching a size class $b$, where
no object from a smaller size class has been encountered. 

To flush the size classes $i\geq b$ at time $t$, first calculate $V_t(i)$
for all $i \geq b$.  The goal is to redistribute these size classes to
take space at most $S = (1+\epsilon')\sum_{i\geq b}V_t(i)$, i.e.,
space $V_t(i)$ for the $i$th payload segment and $\floor{\epsilon'
  V_t(i)}$ for the $i$th buffer, while moving all objects from buffers
into payload segments.

A flush can be implemented to include at most two moves per object in
the flushed size classes.  
\begin{enumerate}
\item 
First, identify the new array suffix of size
$S$ to accommodate payload and buffer segments.  Temporarily move all
objects from buffer segments to empty space immediately after this
suffix (or after the current suffix, if the current suffix is longer due to deletes),
removing any dummy delete records from buffers. 
These objects make up the overflow segment.
This first step
increases space usage by at most $\sum_{i\geq b} \epsilon' V_t(i)$.
\item 
Next, iterate over payload segments from smallest to largest, moving
objects as early as possible, thus removing any gaps left by deleted
objects or emptied buffers.  At the end of this step, all the objects
are packed as far left as possible with no gaps, beginning at the
start of region $b$.  
\item Then, iterate over payload segments from largest
to smallest, moving each object to its final destination in the
redistributed array (which is no earlier than its current location).
The final destination can be determined by looking at the values
$\set{V_t(i)}$; this step reintroduces gaps to accommodate any
not-yet-placed  objects in the overflow segment 
and the empty size-$\floor{\epsilon'  V_t(i)}$ buffers.  
\item Finally, iterate over all objects in the overflow segment,
placing them in their final destinations at the end of the
appropriate payload segments.
\end{enumerate}

\subheading{Analysis}
The proof of \thmref{makespan} follows from
\lemreftwo{makespan-space}{makespan-cost} given below, by fixing
$\epsilon' = \Theta(\epsilon)$ appropriately.  \lemref{makespan-space}
states that the space used is $1+O(\epsilon')$ times the optimal space
usage. 
 \lemref{makespan-cost} states that the reallocation cost is no
worse than $O((1/\epsilon')\log(1/\epsilon'))$ times the optimal
reallocation cost.  

\newcommand{\makespanSpaceStatement}{After 
processing the first $t$ allocation/deallocation requests, 
the space used by the storage-reallocation algorithm is
$(1+O(\epsilon'))V$, where $V = \sum_i V_t(i)$.}

\begin{lemma}\lemlabel{makespan-space}
\makespanSpaceStatement
\end{lemma}

\begin{proof}
  Let $f_i \leq t$ be the previous time the $i$th buffer was flushed.
  The space used by the buffers and payload segments is at most
  $(1+\epsilon')\sum_i V_{f_i}(i)$ by construction, and it may grow to
  $(1+O(\epsilon'))\sum_i V_{f_i}(i)$ during the present buffer flush.
  
  To prove the lemma, we need only bound the difference between 
 $\sum_i V_{f_i}(i)$ and  $\sum_i V_t(i)$.  The difference is accounted for by those objects in buffers
  (including delete records), which amount to at most an
  $\epsilon'\sum_iV_{f_i}(i)$ total volume of objects.  Thus, we have
  $\abs{\sum_i V_t(i)-\sum_i V_{f_i}(i)} \leq \epsilon'\sum_i
  V_{f_i}(i)$.
  
  The worst-case-ratio overhead occurs when all buffered objects are
  deletions, in which case  $\sum_i V_t(i) \geq
  (1-\epsilon')\sum_i V_{f_i}(i)$.
  
  Thus, at most $(1+O(\epsilon'))\sum_i V_{f_i}(i)$ space
  stores at least $(1-\epsilon')\sum_i V_{f_i}(i)$ active objects.
  Observing that $(1+O(\epsilon'))/(1-\epsilon') = 1+O(\epsilon')$ for
  $\epsilon' \leq 1/2$ completes the proof.
\end{proof}

\begin{lemma}\lemlabel{makespan-cost}
  For monotonically increasing, subadditive cost functions $f$, 
  the amortized cost of inserting
  or deleting an object of size $w$ is
  $O\left(f(w)\cdot(1/\epsilon)\log(1/\epsilon)\right)$.
\end{lemma}

\begin{proof}
  Consider a buffer-flush operation, and let $b$ be the boundary size
  class (i.e., all size classes $i\geq b$ have their buffers flushed).
  There are two cases:

  Case 1: The $i$th buffer contains $\Omega(\epsilon' V(i))$ volume of
  objects, for concreteness, say at least $\epsilon' V(i) / 2$ volume.

  Case 2: The $i$th buffer is \defn{underfull}, i.e., contains less than
  $\epsilon' V(i) / 2$ volume. Case 2 occurs
  because of roundoff. Specifically, $\epsilon' V(i)$
  may not be large enough to accommodate even one object in
  size-class $i$.
  
  We first deal with Case~1. We need to show that the initial
  allocation cost of objects in the buffer is sufficient to pay for
  the reallocation cost of objects in the payload segment. Since the
  objects in the buffer belong to the $i$th or earlier size classes,
  they can each have size at most $2^i$.  The cost per unit size, 
  $f(x)/x$, is nonincreasing, so the
  cost of allocating the objects in the buffer is at least
  $\Omega((f(2^i)/2^i)(\epsilon' V(i)))$. Since $f$ is subadditive, we
  have $f(2^i) = O(f(2^{i-1}))$, which implies that this buffer cost
  is at least $\Omega((f(2^{i-1})/2^{i-1})(\epsilon' V(i)))$.  If we
  charge each buffered object for $\Theta(1/\epsilon')$ reallocations,
  it follows that we can afford the total cost of at most
  $(f(2^{i-1})/2^{i-1})V(i)$ to reallocate all objects in the payload
  segment.  This case is completed by observing that each object is only flushed once: after an
  object moves to the payload segment, it stays there until it is
  deallocated.

  We next deal with Case~2, where buffer~$i$ is underfull.
  Buffer~$i$ participates in the buffer-flush operation because some
  object belonging to size class $i' \leq i$ is placed in some buffer
  for size class $j > i$.  We charge that object for flushing any
  underfull buffers between size class $i'$ and size class $j$. (There
  may be many such objects, which only decreases the cost per
  object---we pessimistically charge only a single object.)

  The main question, then, is: what is the maximum reallocation cost
  due to underfull buffers that can be charged against an object in
  size-class $i'$?  Size-class $i$ may only be charged against the
  object if $2^{i'} > \epsilon' V(i)/2$.  This implies that $V(i) =
  O(2^{i'}/\epsilon')$, and hence the cost of moving every object in
  size-class $i$ is at most $O(1/\epsilon')$ times the cost of
  allocating a single object in size-class $i'$.  Because each
  successive size class doubles in size, and a size class only has a
  payload segment (and buffer segment) if there is at least one object
  in the size class, only the $O(\log(1/\epsilon'))$ nearest size
  classes may satisfy $2^{i'} > \epsilon' V(i)/2$---in particular,
  $\epsilon' 2^{i'+\ceil{\log(1/\epsilon')} + 1}/2 \geq 2^{i'}$, and
  hence if any larger size-class is underfull, it will not be
  ``skipped over'' by an object in size-class $i'$.

  To conclude, buffered objects in size-class $i'$ may be charged for
  $O(1/\epsilon')$ reallocations in $O(\log(1/\epsilon'))$ different size
  classes, for a total cost of $O((1/\epsilon')\log(1/\epsilon'))$
  allocations.
\end{proof}

\subsection*{Corollary: Defragmenting/Sorting} 
\seclabel{sort}

A corollary of cost-oblivious storage reallocation is a cost-oblivious
defragmentation algorithm, i.e., a cost-oblivious algorithm for
sorting the objects while simultaneously respecting constraints on the
space usage.

We first compare with na\"\i{}ve defragmentation.  If $2V$ working
space is allowed, then defragmentation is trivial with two movements
per object.  First pack the objects into the rightmost $V$ space,
using one move per object.  Then place each object directly in its
final destination within the leftmost $V$ region of space.

The following theorem shows that defragmentation is possible even using
 $ (1+\epsilon)V+ \Delta$  space by applying cost-oblivious storage
reallocation as a black box.

\begin{theorem} 
For any $0 < \epsilon \leq 1/2$ there exists
  a cost-oblivious defragmentation algorithm that takes as input
  (1)~an arbitrary comparison function, (2)~a set of objects with
  volume $V$, and (3)~a current allocation of the objects using space
  at most $(1+\epsilon)V$.  The algorithm sorts the objects according
  to the comparison function, subject to:
  \begin{itemize}
  \item the total space usage at any time never exceeds
    $(1+\epsilon)V+\Delta$ space, and

  \item the total cost is at most $O((1/\epsilon)\log(1/\epsilon))$
    times the cost to allocate all of the objects. 
  \end{itemize}
\end{theorem}

\begin{proof}
  First crunch the objects into the rightmost $V$ space, leaving a
  size-$\floor{\epsilon V}$ prefix of the array empty.  We reserve
  this prefix to run the cost-universal storage-reallocation
  algorithm.  Starting with the leftmost object in the suffix, remove
  it from the suffix, store it temporarily in the $\Delta$ additional
  space, and then insert it into the prefix using cost-universal
  storage reallocation. Since the storage reallocation guarantees at
  most $(1+\epsilon)W$ space usage, for $W$ total volume of objects in
  the prefix, at no point does the prefix of size at most $(1+\epsilon)W$ 
  overlap the suffix of size $(V-W)$.  When this process completes, the
  suffix is empty and all objects are in the cost-universal-storage
  data structure.

  Next,  move elements back to the suffix in reverse sorted order.
  Specifically, delete each object from the prefix (using the
  cost-universal storage-reallocation algorithm), which 
  compacts the space used, and place the object just before its
  successor in the suffix.  Again, at any time, if $W$ is the remaining
  volume of objects in the prefix, the prefix uses at most
  $(1+\epsilon)W$ space, and the suffix uses exactly $V-W$ space, so
  the prefix does not overlap the suffix.
\end{proof}

Note that the additional $\Delta$ working space is unavoidable when
reallocating large objects. To see this, consider a single size-$\Delta$
object.  This object cannot be moved unless the target location is not
overlapping with the original location.  That is, if we have less than
$2\Delta$ space to work with, the object can never be moved as every
target location overlaps its current location.

\section{Footprint Minimization in a  Database Context}
\seclabel{db}

This section extends the storage-reallocation algorithm to take into
account issues that arrise in databases: durability and
blocking.  To provide durability, we extend the algorithm to work with
a checkpointing mechanism.  Specifically, we show how to complete a
buffer flush in $O(1/\epsilon)$ checkpoints.
During a flush, the memory footprint increases by an additive $\Delta$
term, up to $(1+\epsilon)V + \Delta$, where $V$ is the total length of
all active objects, and $\Delta$ is the length of the longest object.
The additive $\Delta$ is unavoidable due to the fact that when a large
object is moved, its new location cannot overlap its old location.

To prevent updates from blocking for too long, we present a (partially)
deamortized version.  The deamortized data structure
has the same amortized reallocation cost and memory footprint as
the original, but it also has a worst-case reallocation cost
of $O((1/\epsilon)w f(1) + f(\Delta))$ for inserting/deleting a
size-$w$ object.  That is, on each update, the total length of jobs
reallocated is roughly proportional to the size of the object being
inserted/deleted.  Viewed differently, the deamortized bound shows that the
desired footprint bound can be maintained with nonblocking updates,
%
as long as the updates arrive infrequently enough that the previous
update has been handled, that is, as long as the previous update of
size $w$ is followed by a gap of size $\Omega((1/\epsilon)w f(1) +
f(\Delta))$.

\subsection{Overview of the Checkpointing Model}

Recall that moving an object updates the map that is maintained
between logical and physical addresses.  From time to time, and
specifically during a checkpoint, this map is written to disk, so that
a database that is recovering from a crash has access to the updated
map.  Suppose an object is reallocated.  Then the map must be updated.
But if a crash occurs before the next checkpoint, the updated map will
not be available to the database on recovery.  Therefore, we must
maintain two copies of the data---at the old and new locations---until
the next checkpoint has completed.  Only then is it safe to assume
that the database knows, in a durable fashion, the new location of the
data.  

The consequence for designing a reallocator is that from time to time,
the database will perform a checkpoint, and all the space that was
freed since the last checkpoint will become available.  The
requirement that moved data reside in two locations until the next
checkpoint means that the system needs an enforcement mechanism.  This
mechanism guarantees that 
if our
algorithm would like to write to a freed but not checkpointed location
it will block.  Therefore, a reallocation algorithm is better if it
requires fewer checkpoints to compete.  For example, if we were to
write the data to completely new locations, the algorithm would not
block on any checkpoints, because we would not be reusing any space.
However, the competitive ratio of the footprint would be at least
two.  We show below that we can achieve our bound of $(1+\epsilon)$
competitive ratio while blocking on at most $O(1/\epsilon)$ checkpoints.

The timing of checkpoints is dependent on many considerations
beyond the needs for reallocation, so we assume that checkpoints are
initiated by the system, rather than our algorithm.  There are other
models of checkpointing, such as log-trimming through incremental
checkpointing.  A complete treatment of checkpointing is beyond the
scope of this paper, though it would be interesting to see how
different types of checkpointing interact with reallocation.

\subsection{Flushing with Checkpoints}

The goal of the flush here is identical to that in \secref{simplealloc},
but the implementation details differ to accommodate the checkpointing
model.  Namely, the space used increases by an additive $\Delta$, and
the flush itself proceeds in several rounds with checkpoints in
between.  Another improvement here is that an inserted element
gets inserted \emph{before} the flush completes, whereas in
\secref{simplealloc} we assumed for simplicity that the insert blocks
until the flush completes.  The memory footprint at the end of the
flush is identical to that of the previous algorithm.


\subheading{Inserting (allocating) and deleting (deallocating)} Since
objects only move during a buffer flush, the insert and delete
procedure is almost identical to \secref{simplealloc}.  The only
difference here is that we insert the object \emph{before} triggering
a flush.

To insert an object, place it in the appropriate buffer segment as
before.  If there is insufficient space to place the object in any
following buffer segment, place it at the end of the last buffer
segment (filling and exceeding the buffer capacity) and trigger a flush.  When
deleting an object, insert a dummy delete request as in
\secref{simplealloc}.  If this delete request would overflow the last
buffer, then trigger the flush without using space for the dummy
delete request.

\subheading{Buffer flush}
A flush proceeds as follows.  First identify the boundary size class
$b$ as before.  Recall that the flush proceeds on size classes $i\geq
b$.  Let $L$ denote the endpoint of the last object \emph{before} the
insert/delete that triggers the flush, i.e., if the total space is $S$
including a newly inserted size-$w$ object, then
$L=S-w$. (Note that this detail of subtracting off the newly inserted
  object is important to obtain a space usage of $(1+\epsilon)V +
  \Delta$ throughout the flush rather than $(1+\epsilon)V +
  O(\Delta)$.)
Let $L'$ be the desired memory footprint after the flush, but
subtracting off the size of any flush-triggering
insert; similar to the procedure for ``$S$'' discussed in
  \secref{simplealloc}, $L'$ can be calculated by first computing
  $\sum_{i \geq b} (V_t(i) + \floor{\epsilon' V_t(i)})$.  That is, if
the final data structure should take $S'$ space after the flush, then
$L' = S'-w$, where $w$ is the size of the last insert if the flush was
triggered by an insert.  Let $B$ be the total space occupied by the
buffers involved in the flush.  Move all objects from buffer segments
$i\geq b$ to the end of the array, starting from location
$(\max\set{L,L'}+B+\Delta)$.  The important observation here is that
$L+\Delta$ exceeds the location of the newly inserted object, so none
of the target locations overlap any of the current objects.  Hence
all of these movements can be performed within a single checkpoint.
The order in which the buffered objects are moved does not matter.
This step of the flush is similar to \secref{simplealloc}, except the
starting location is up to $B+\Delta$ slots later in the array.

Next, iterate over payload segments from largest to smallest, moving
objects as late as possible in the array ending at location
$(\max\set{L,L'}+B+\Delta)$.  After this step, flushed payload
segments are packed as late as possible before location
$(\max\set{L,L'}+B+\Delta)$, and flushed buffer segments (including the
newly inserted object) are packed as early as possible after
$(\max\set{L,L'}+B+\Delta)$.

This payload-packing step, however, moves objects to locations in the
array that may have previously been occupied, which would violate the
checkpointing model.  Instead, break these movements into phases with
checkpoints between each phase. Move as many objects
as possible before exceeding $B+\Delta$ volume in each phase.  Since the largest
object has size $\Delta$, the minimum amount moved is $B+1$.
As we
shall prove, the movements within a phase do not overlap, and the
total number of phases is $O(1/\epsilon')$.  Aside
from checkpointing, this step differs from the version in
\secref{simplealloc} in that objects are packed later in the array
rather than earlier, and hence the movements iterate from
largest-to-smallest size class rather than smallest-to-largest.  The
reason for this change is to take advantage of the $B+\Delta$ working
space available at the end of the region.

Next, iterate over payload segments from smallest to largest, moving
the objects exactly where they should go in the array.  This step,
again, may move objects to space that was previously occupied, so we
again break it into phases consisting of the next $B+1$ to $B+\Delta$
target locations with a checkpoint following each phase.

Finally, move the buffered elements to their target locations. Since
all buffered elements are currently located after
$(\max\set{L,L'}+B+\Delta)$, and all target locations are before
$L'+\Delta$, none of these movements overlap, and they can be
performed within a single checkpoint.

\subheading{Analysis} Note that the number of reallocations is similar
to that in \secref{simplealloc}, with the only difference being one
reallocation for the flush-triggering item.  Hence the reallocation
cost of \lemref{makespan-cost} holds for this version of the
algorithm.  The space used after a flush completes is also identical
to \secref{simplealloc}.  It remains to prove three facts: 1) the
space used during a flush is $(1+O(\epsilon'))V+\Delta$ where $V$ is
the total volume of active jobs, 2) the object movements between
checkpoints only move objects to nonoverlapping locations, and 3) the
number of checkpoints is $O(1/\epsilon')$ per flush.

\begin{lemma}
  While processing any allocation/deallocation request, the total
  footprint used by
  the algorithm is at most $(1+O(\epsilon'))V + \Delta$, where $V$
  denotes the total volume of all currently active
  objects.
\end{lemma}
  \newcommand{\Vb}{V_{\scriptsize\id{\rm before}}}
  \newcommand{\Va}{V_{\scriptsize\id{\rm after}}}
  \newcommand{\Sb}{S_{\scriptsize\id{\rm before}}}
  \newcommand{\Sa}{S_{\scriptsize\id{\rm after}}}

\begin{proof}
  Let $\Vb$ and $\Va$ denote the total volume of objects before and
  after the operation, respectively.  Let $\Sb$ and $\Sa$ denote the
  total space of the data structure before and after the operation,
  respectively.  According to \lemref{makespan-space}, we have $\Sb
  \leq (1+O(\epsilon'))\Vb$, and $\Sa \leq (1+O(\epsilon')) \Va$.  The
  question is what happens during the operation, notably during a
  flush operation.

  Suppose the flush is triggered by a size-$w$ insertion.  The volume
  during the flush is thus $V = \Vb + w = \Va$.  The space used to
  store all buffered objects, including the newly inserted object, is
  at most $w+B$, where $B$ is the total amount of space devoted to
  buffers before the flush.  Note that since the buffers are sized to
  less than an $\epsilon'$ fraction of the total space, we have $B
  \leq \epsilon'
  \Sb$.\\
  Case 1: $\Sb \geq \Sa$.  Then these objects are written at an offset
  of $(\Sb + B + \Delta)$, meaning that the total space during the
  flush is at most
  \begin{eqnarray*}
    &&(\Sb + B + \Delta) + (w+B) \\
    &\leq& (1+2\epsilon')\Sb + w +
    \Delta \hfill\quad\quad \text{\footnotesize \quad \sf // upper bound on $B$}\\
    &\leq& (1+2\epsilon')\left[(1+O(\epsilon'))\Vb\right] + w + \Delta
    \hfill \text{\footnotesize \quad \sf // \lemref{makespan-space}} \\
    &\leq& (1+O(\epsilon'))\Vb + w + \Delta \hfill \quad\!\! \text{\footnotesize \quad \sf // larger const in big-O}  \\
    &\leq& (1+O(\epsilon'))(\Vb + w) + \Delta \\
    &=& (1+O(\epsilon'))V+\Delta \ .
  \end{eqnarray*}
  Case 2: $\Sb < \Sa$. Then these objects are written at an offset of
  $(\Sa-w) + B + \Delta$.  And the total space during the
  flush is at most $\Sa + 2B + \Delta \leq
  (1+O(\epsilon'))\Va + \Delta = (1+O(\epsilon'))V +
  \Delta$, where the steps follow from analogous steps in Case~1.

  In the case of a deletion, the argument is similar, except $w$
  becomes 0 in all the expressions, and $V=\Vb$ throughout the flush.
  That is, the deleted object is considered active until the flush completes.
\end{proof}

\begin{lemma}
  During a single phase of object movements between two checkpoints,
  all object starting locations are disjoint from all object ending
  locations.
\end{lemma}
\begin{proof}
  First, consider the payload-packing step, where payload segments are
  packed to the right.  At the start of the $j$th phase, let $\ell_j$
  denote the last cell occupied by the payload segments that have yet
  to be packed, and let $r_j$ denote the first occupied cell later
  than $\ell_j$.  We claim that at the start of each phase $r_j \geq
  \ell_j + B + \Delta$, which we shall prove by induction.  If true,
  the claim implies disjointness: if the space between $r_j$ and
  $\ell_j$ is at least $B + \Delta$, then we can pack up to $B+\Delta$
  volume of jobs in front of $r_j$ during the $j$th phase before
  overlapping the ending position of jobs at $\ell_j$.

  We prove the claim by induction.  The claim holds initially because
  $\ell_0 \leq L$, and $r_0 \geq L + B + \Delta$.  For the inductive
  step, observe that if $X$ volume of objects are moved in phase $j$,
  then $\ell_{j+1} \leq \ell_j - X$, and $r_{j+1} = r_j - X$.
  Combined with the inductive assumption that $r_j \geq \ell_j + B +
  \Delta$, we get $r_{j+1} \geq (\ell_j + B + \Delta) - X \geq
  ((\ell_{j+1} + X) + B + \Delta) - X = \ell_{j+1} + B + \Delta$.

  We next consider the unpacking step, where the payload segments are
  moved to their final positions.  Let $\ell_j$ denote the last cell
  occupied by unpacked payload segments at the start of the $j$th
  phase of movements, and let $r_j$ denote the first cell occupied by
  the yet-to-be unpacked payload objects.  We claim that $\ell_j + B +
  \Delta \leq r_j$ (but this time we shall prove it by contradiction).
  If the claim holds, then we can afford to increase $\ell_j$ by
  $B+\Delta$ in each phase without violating the disjointness.

  To prove the claim, suppose for the sake of contradiction that
  $\ell_j > r_j - B - \Delta$, and let $X$ be the total volume
  remaining in the packed region.  Then the final position of the
  last payload segment can end no earlier than $\ell_j + X > r_j + X - B -
  \Delta$ after the unpacking, and hence the space desired by these
  payload segments is at least $L' > (r_j + X) - B - \Delta$.  We also
  have $r_j+X = \max\set{L',L}+B+\Delta$ is the offset at which the
  buffered objects were moved, which we simplify to $r_j+X \geq L' + B
  + \Delta$. Combining these two facts, we get $L' > (r_j + X) - B -
  \Delta \geq (L' + B + \Delta) - B - \Delta = L'$, i.e., $L' > L'$,
  which is a contradiction.
\end{proof}

\begin{lemma}
  The number of checkpoints occurring during a flush is
  $O(1/\epsilon')$.
\end{lemma}
\begin{proof}
  The checkpoints are dominated by the packing and unpacking steps.
  Let $P(i)$ denote the total space of the $i$th payload segment at
  the time of the flush, i.e., the volume of jobs that were in this
  size class the last time a flush occurred. Then the total size of
  flushed buffers is $B = \sum_{i\geq b} \floor{\epsilon'P(i)}$, and
  the total space of the region being flushed is $S=\sum_{i\geq b}
  (P(i) + \floor{\epsilon' P(i)})$.  Since each movement phase does
  more than $B$ work, showing that $B = \Omega(\epsilon'S)$ would be
  sufficient. The only difficulty is the floor in the expression, so
  we shall consider the case of large $P(i)$ and small $P(i)$
  separately.

  Case 1: sufficiently large $P$.  More precisely, suppose $B =
  \sum_{i\geq b} \floor{\epsilon' P(i)} \geq \sum_{i\geq b}\epsilon'
  P(i)/2$.  Then $B = \Omega(\epsilon' S)$, since $\sum_{i\geq b}P(i)
  \geq S/2$ for $\epsilon < 1$.

  Case 2: small $P$.  Suppose $B < \sum_{i\geq b}\epsilon' P(i)/2$.
  Note that $B =\sum_{i\geq b}\floor{\epsilon' P(i)} \geq \sum_{i\geq
    b}\epsilon' P(i) - \Theta(\log\Delta)$, since there are only
  $\Theta(\log \Delta)$ size classes.  It follows that $B < \sum_{i
    \geq b}\epsilon' P(i)/2$ implies $\sum_{i\geq b}\epsilon' P(i) =
  O(\log \Delta)$, and hence $S = O((1/\epsilon')\log\Delta)$.  The
  algorithm tries to move as many objects as it can until exceeding
  $B+\Delta$ volume, and hence every consecutive pair of phases moves
  at least $\Delta/2 = \Omega(\log\Delta) = \Omega(\epsilon'S)$ volume.
\end{proof}

\subsection{Deamortizing the Data Structure}

As described so far, the data structure is amortized---the average
reallocation cost per update is low, but on some updates \emph{every}
active object may need to be reallocated (i.e., when all size classes
are involved in a flush).  This section improves the worst-case
reallocation cost of a size-$w$ update to
$O((1/\epsilon)wf(1)+f(\Delta))$, without hurting the amortized update
cost or the maximum footprint.

Note that the deamortization described here builds on
the checkpointing modification, yielding a worst-case $O(1/\epsilon)$ checkpoints
per operation.

\subheading{Modifications to the algorithm}
The main idea of our deamortization is that if a buffer flush performs
a total of $X$ reallocations by volume, then this work is spread
across the subsequent $\epsilon' X$ updates by volume.  The question,
however, is where to place new objects that are inserted during a
flush. If, for example, an insert could trigger a smaller flush while
a larger flush is still ongoing, that would present even more
challenges.  We tackle these problems by adding two more buffers to
the data structure and modifying the flush, which serve to avoid the
issue of nested flushes.

Augment the data structure to include one size-$\floor{\epsilon' V_f}$
buffer, called the \defn{tail buffer}, following all the size-class
segments, where $V_f$ is the total volume of all jobs active at the
start of the previous buffer flush.  The tail buffer is like any other
buffer: objects are only placed in the tail buffer if all earlier
buffers are too full, and a buffer flush is only triggered once the
tail buffer becomes full.  The point of the large tail buffer is to
enable the flush to complete before triggering another flush.

When a flush is triggered, calculate the desired space and the
temporary working space as before; however, the space is slightly
  larger now due to the $\floor{\epsilon' V}$ space necessary for the
  tail buffer.  We treat all space immediately following the
temporary working space as another buffer called the \defn{log}.

The flush process resembles the previous flush process (with or without
checkpointing), except that:
\begin{enumerate}
\item Objects may be inserted/deleted during a flush.  These updates
  are placed at the end of the log.
\item The work of the flush is spread across these subsequent updates.
  Specifically, on an insertion/deletion of a size-$w$ object, perform
  (just over) the next $(4/\epsilon')w$ steps of the flush by volume.
  Since a fractional object cannot be moved, the amount of volume
  processed may be as high as $(4/\epsilon')w+\Delta$.
\item There is an extra phase at the end.  During this phase, all
  objects in the log are moved to their appropriate buffers, i.e.,
  they are re-inserted/re-deleted.  This phase proceeds in order from
  the beginning of the log.  Updates may continue to be recorded at
  the end of the log during this phase.  Since the volume moved is
  significantly larger than the size of the update, the
  re-insertion/re-deletion will eventually ``catch up'' to the end of
  the log, at which point the flush terminates and the log
  disappears.
\end{enumerate}

\subheading{Analysis}
To show correctness of the new flush protocol, we argue that the log
is drained completely before another flush gets triggered, i.e.,
before the tail buffer fills. Note that if the last update during a
flush involves a large object, that update may finish the previous
flush and trigger the next one. The point is only that the tail buffer
cannot overflow before that time.

\begin{lemma}\lemlabel{logsize}
  Let $V_f$ be the total volume of active objects at the time a flush
  is triggered. For any $\epsilon' < 1$, the flush completes by the
  time the subsequent volume of updates first exceeds $\epsilon'
  V_f$.
\end{lemma}
\begin{proof}
  In the worst case, a flush may move every object twice.
  Specifically, the buffered elements are moved out of the buffers
  temporarily, then to their final location. Similarly, the payload
  segments are packed once and then unpacked to their final location.
  The total volume of reallocations of preexisting elements is thus at
  most $2V_f$.  (Any delete records do not have to be reallocated;
  these are just destroyed.)

  It follows that by the time $(\epsilon'/2) V_f$ volume of updates
  are logged, all preexisting elements have been moved to their final
  locations.  But this analysis does not take into account the
  elements logged during the flush.  The next $(\epsilon'/2)V_f$
  volume of updates more than suffice to move all objects from the log
  to a buffer.
\end{proof}

The following lemmas bound the space and reallocation costs of the
updated algorithm.

\begin{lemma}
  After each allocation/deallocation request is processed, the
  total space used by the data structure is at most $(1+O(\epsilon'))V
  + \Delta$, where $V$ denotes the total volume of all currently
  active objects.  If a flush is not in progress, this space improves
  to $(1+O(\epsilon'))V$.
\end{lemma}
\begin{proof}
  The only significant difference in space between this algorithm and
  the amortized one is the tail buffer and the log.  The tail buffer
  has size at most $\epsilon' V_f$, where $V_f$ was the volume at the
  last flush.  According to \lemref{logsize}, the log also has size at
  most $\epsilon' V_f$.  Combined, the total increase to space is an additive
  $O(\epsilon')V_f$.  To complete the argument, we need only argue
  that $V_f = \Theta(V)$, where $V$ is the current volume of
  active jobs, which is done in the proof of \lemref{makespan-space}.
\end{proof}

\begin{lemma}
  For subadditive cost function $f$, the amortized cost of inserting
  or deleting a size-$w$ object is
  $O(f(w)\cdot(1/\epsilon')\log(1/\epsilon'))$.  Moreover, the worst-case
  cost of an insert or delete is $O((1/\epsilon')w f(1) + f(\Delta))$.
\end{lemma}
\begin{proof}
  The worst-case upper bound follows by construction.  The algorithm
  only reallocates $(4/\epsilon')w$ volume of objects per update, plus
  up to one last object to exceed this volume.  In the worst case,
  these objects are size-1 objects except the last which is
  size-$\Delta$, for a total cost of $O((1/\epsilon)w f(1) +
  f(\Delta))$.

  As for the amortized bound, adding a larger buffer to the data
  structure only improves the amortized cost.  Specifically, the proof of
  \lemref{makespan-cost} relied on \emph{lower bounding} the volume of
  buffered objects, so the same analysis applies once an object is
  placed in a buffer.  The deamortized data structure has an
  additional reallocation for each object that is placed in the log,
  moving it from the log to a buffer, but this only occurs once per
  object.
\end{proof}

\subheading{Lower bound on worst-case cost}
Note that $\Omega(f(\Delta))$ is a lower bound on the
worst-case reallocation cost when maintaining a $(1+\epsilon) V$
footprint size, as exhibited by the following lemma.  It is not
obvious whether $\Omega(w f(1))$ is also a lower bound on the
worst-case reallocation cost of any algorithm.  If so, then our
deamortized structure's worst-case cost would be asymptotically
optimal for constant~$\epsilon$.  Although not a general lower bound,
an $\Omega(w f(1))$ worst-case cost appears to be unavoidable for any
algorithm that stores ``enough'' small objects after large
objects. (And storing objects out of order in this way seems
  crucial for obtaining a \emph{cost-oblivious} algorithm.)
Informally, deleting a size-$w$ object leaves a large hole in the
array.  To maintain the desired footprint, this hole must be filled by
later objects. If all later objects are small (size-1), then a
size-$w$ delete may cause $\Omega(w)$ size-1 objects to move.

\begin{lemma}
  For any reallocation algorithm that maintains a footprint of
  $(1+1/2)V$ and subadditive cost function~$f$, there exists an update
  sequence such that at least one update has a reallocation cost of
  $\Omega(f(\Delta))$.  This lower bound applies even if the
  reallocation algorithm knows~$f$ and the full update sequence.
\end{lemma}
\begin{proof}
  Here is the sequence. First insert one size-$\Delta$ object.  Then
  insert $\Delta$ size-1 objects.  Then delete the size-$\Delta$
  object.  There are two cases to show the lower bound.\\
  Case 1: some small-object insertion causes the large object to be
  reallocated.  Then that insert has a reallocation cost of
  at least $f(\Delta)$.\\
  Case 2: the large object does not get reallocated.  Then the large
  object must end before position $(3/2)\Delta$ to achieve the
  footprint bound, and hence there must be at least $\Delta/2$ small
  objects appearing after the large one.  When deleting the large
  object, those small objects must move in order to restore the
  $(3/2)\Delta$ footprint bound.  Hence the cost of deleting the large
  objects is $\Omega(\Delta \cdot f(1)) \subset \Omega(f(\Delta))$ for
  subadditive~$f$.
\end{proof}

{
\bibliographystyle{ACM-Reference-Format-Journals}
\bibliography{./reallocation}
}

\end{document}